\RequirePackage{silence}
\documentclass[a4paper,fontsize=11pt]{scrartcl}
\usepackage[UKenglish]{babel}
\usepackage[utf8]{inputenc}
\usepackage{etex}
\usepackage{microtype}
\usepackage[UKenglish]{babel}
\usepackage[libertine]{newtxmath}
\usepackage[tt=false, type1=true]{libertine}
\usepackage{csquotes}

\usepackage{geometry}
\usepackage{ifthen}
\newif\iflong
\longtrue

\usepackage{amsmath,amssymb}
\usepackage{stmaryrd}
\usepackage[standard,amsmath,thmmarks]{ntheorem}

\theoremstyle{break}
\theorembodyfont{\normalfont}
\theoremsymbol{}
\renewtheorem{definition}[envs]{Definition}

\theorembodyfont{\itshape}
\renewtheorem{proposition}[envs]{Proposition}
\renewtheorem{corollary}[envs]{Corollary}
\renewtheorem{theorem}[envs]{Theorem}
\renewtheorem{lemma}[envs]{Lemma}

\usepackage{doi}
\usepackage[anythingbreaks]{breakurl}

\makeatletter
\renewtheoremstyle{plain}%
{\item[\hskip\labelsep \theorem@headerfont ##1\ ##2 \theorem@separator]}%
{\item[\hskip\labelsep \theorem@headerfont ##1\ ##2\ \normalfont({\sffamily##3})  \theorem@separator]}
\makeatother

\RedeclareSectionCommand[indent=0pt]{subparagraph}

\usepackage{microtype}

\usepackage{booktabs}
\usepackage{bold-extra,bm}

\newcommand\enumproblem[4]{%
\begin{center}
\begin{tabular}{lp{9cm}}\toprule
\textsf{\bfseries Problem:}& #1 \\\midrule
\textsf{\bfseries Input:}& #2.\\
\textsf{\bfseries Parameter:}& #3.\\
\textsf{\bfseries Output:}& #4.\\\bottomrule
\end{tabular}
\end{center}
}

\newcommand{\funcparaproblemdef}[4]{%
\begin{center}
\begin{tabular}{lp{9cm}}\toprule
\textsf{\bfseries Problem:}& #1 \\\midrule
\textsf{\bfseries Input:}& #2.\\
\textsf{\bfseries Parameter:}& #3.\\
\textsf{\bfseries Task:}& #4.\\\bottomrule
\end{tabular}
\end{center}
}

\usepackage[nofillcomment,scleft,vlined,ruled,linesnumbered]{algorithm2e}
\SetKwSty{mykwfn}

\newcommand{\logicClFont}[1]{\mathrm{#1}}

\newcommand{\problemFont}[1]{\textsc{#1}}

\newcommand{\complClFont}[1]{\mathsf{#1}}
\newcommand{\N}{\protect\ensuremath{\mathbb N}\xspace}

\renewcommand{\P}{\protect\ensuremath{\complClFont{P}}\xspace}
\newcommand{\PTime}{\P}
\newcommand{\NP}{\protect\ensuremath{\complClFont{NP}}\xspace}
\newcommand{\FP}{\protect\ensuremath{\F\P}\xspace}
\newcommand{\FNP}{\protect\ensuremath{\F\NP}\xspace}

\newcommand{\FPT}{\protect\ensuremath{\complClFont{FPT}}}
\newcommand{\para}{\protect\ensuremath{\complClFont{para\text-}}}

\newcommand{\Wone}{\protect\ensuremath{\W1}}
\newcommand{\W}[1]{\protect\ensuremath{\complClFont{W}[#1]}}

\newcommand{\enumP}{\protect\ensuremath{\complClFont{P\text-enum}}}
\newcommand{\outputP}{\protect\ensuremath{\complClFont{OutputP}}}
\newcommand{\delayP}{\protect\ensuremath{\complClFont{DelayP}}}
\newcommand{\capincP}{\protect\ensuremath{\complClFont{CapIncP}}}
\newcommand{\incP}{\protect\ensuremath{\complClFont{IncP}}}
\newcommand{\totalP}{\protect\ensuremath{\complClFont{\outputP}}}

\newcommand{\outputFPT}{\protect\ensuremath{\complClFont{OutputFPT}}}
\newcommand{\FPTenum}{\protect\ensuremath{\complClFont{EnumFPT}}}
\newcommand{\delayFPT}{\protect\ensuremath{\complClFont{DelayFPT}}}
\newcommand{\incFPT}{\protect\ensuremath{\complClFont{IncFPT}}}
\newcommand{\capincFPT}{\protect\ensuremath{\complClFont{CapIncFPT}}}

\newcommand{\totalFPT}{\protect\ensuremath{\complClFont{TotalFPT}}}
\newcommand{\enum}{\protect\ensuremath{\problemFont{Enum}\text-}}
\newcommand{\anotherSol}{\protect\ensuremath{\problemFont{AnotherSol}}}

\newcommand{\AllSol}{\protect\ensuremath{\calS}}

\newcommand{\calA}{\mathcal A}
\newcommand{\calS}{\mathcal S}

\newcommand{\Sol}{\protect\ensuremath{\logicClFont{Sol}}}
\newcommand{\F}{{\protect\ensuremath{\complClFont{F}}}}
\newcommand{\TF}{{\protect\ensuremath{\complClFont{TF}}}}

\newcommand{\co}{{\protect\ensuremath{\complClFont{co}}}}

\usepackage{microtype}

\title{Enumeration in Incremental FPT-Time}

\author{Arne Meier\footnote{Funded by the German Research Foundation DFG, project ME 4279/1-2.}}

\date{\small
Institut für Theoretische Informatik, Leibniz Universität Hannover,\\ Appelstrasse 4, 30167 Hannover, Germany\\ \texttt{meier@thi.uni-hannover.de} 
}

\begin{document}

\maketitle

\begin{abstract}
	In this paper, we study the relationship of parametrised enumeration complexity classes defined by Creignou et~al. (MFCS 2013).
	Specifically, we introduce two hierarchies (IncFPTa and CapIncFPTa) of enumeration complexity classes for incremental fpt-time in terms of exponent slices and show how they interleave.
	Furthermore, we define several parametrised function classes and, in particular, introduce the parametrised counterpart  of the class of nondeterministic multivalued functions with values that are polynomially verifiable and guaranteed to exist, TFNP, known from Megiddo and Papadimitriou~(TCS 1991).
	We show that TF(para-NP) collapsing to F(FPT) is equivalent to OutputFPT coinciding with IncFPT.
	This result is in turn connected to a collapse in the classical function setting and eventually to the collapse of IncP and OutputP which proves the first direct connection of classical to parametrised enumeration.	
\end{abstract}

\section{Introduction}
\paragraph{Enumeration.}
In 1988, Johnson, Papadimitriou and Yannakakis \cite{JohnsonPY88} introduced the framework of enumeration algorithms.
In modern times of ubiquitous computing, such algorithms are of central importance in several areas of life and research such as combinatorics, computational geometry, and operations research \cite{af96}.
Also, recent results unveil major importance in web search, data mining, bioinformatics, and computational linguistics \cite{CreignouOS11}.
Moreover, there exist connections to formal languages on enumeration problems for probabilistic automata \cite{DBLP:journals/mst/Strozecki13}.

Clearly, for enumeration algorithms the runtime complexity is rather peripheral and the time elapsed between two outputs is of utmost interest.
As a result, one measures the \emph{delay} of such algorithms and tries to achieve a uniform stream of printed solutions.
In this context, the complexity class $\delayP$, that is polynomial delay, is regarded as an efficient way of enumeration.
Interestingly, there exists a class of \emph{incremental polynomial delay}, $\incP$, which contains problems that allow for enumeration algorithms whose delay increases in the process of computation.
Intuitively, this captures the idea that after printing `obvious' solutions, later in the process it becomes difficult to find new outputs.
More precisely, the delay between output $i$ and $i+1$ is bounded by a polynomial of the input length and of $i$.
Consequently, in the beginning, such an algorithm possesses a polynomial delay whereas later it eventually becomes exponential (for problems with exponential many solutions; which is rather a common phenomenon).
While prominent problems in the class $\delayP$ are the enumeration of satisfying assignments for Horn or Krom formulas \cite{CreignouH97}, structures for first-order query problems with possibly free second-order variables and at most one existential quantifier \cite{DBLP:conf/csl/DurandS11}, or cycles in graphs \cite{rt75}, rather a limited amount of research has been invested in understanding $\incP$.
A well-studied problem in this enumeration complexity class is the task of generating all maximal solutions of systems of equations modulo~2 \cite{DBLP:journals/siamdm/KhachiyanBEGM05}.
Even today, it is not clear whether this problem can be solved with a polynomial delay. 
Other examples for problems in $\incP$ are given by Eiter et~al.~\cite{DBLP:journals/siamcomp/EiterGM03}, or Fredman and Khachiyan \cite{DBLP:journals/jal/FredmanK96}.
Recently, Capelli and Strozecki \cite{cs17} deeply investigate $\incP$ and its relationship to other classical enumeration classes, improving the overall understanding of this class.

\paragraph{Parametrised Complexity.} The framework of parametrised complexity \cite{DowneyFellows13,DowneyFellows99,FlumGrohe06,nie06} allows one to approach a fine-grained complexity analysis of problems beyond classical worst-case complexity.
Here, one considers a problem together with a parameter and tries to achieve deterministic runtimes of the form $f(\kappa(x))\cdot p(|x|)$, where $\kappa(x)$ is the value of the parameter of an instance $x$, $p$ is a polynomial, and $f$ is an arbitrary computable function.
The mentioned runtime is eponymous for the class $\FPT$.
As usually a parameter is seen do be slowly growing or even of constant value \cite{DBLP:journals/jal/AlberFN04}, accordingly, one calls such problems \emph{fixed-parameter tractable}.

A rather large parametrised complexity class is $\para\NP$, the nondeterministic counterpart of $\FPT$, which is defined via nondeterministic runtimes of the same form $f(\kappa(x))\cdot p(|x|)$.
Clearly, $\FPT\subseteq\para\NP$ is true, but essentially $\para\NP$ is widely \emph{not} seen as a correspondent of $\NP$ on the classical complexity side.
In fact, $\Wone$ is the class which usually is used to show intractability lower bounds in the parametrised setting.
This class is part of an infinite $\mathsf{W}$-hierarchy in between the aforementioned two classes.
It is not known whether any of the inclusions of the intermediate classes is strict or not.

\paragraph{Parametrised Enumeration.} Recently, Creignou et~al.\ \cite{cmmsv13,cmmsv17,ckmmov15} developed a framework of parametrised enumeration allowing for fine-grained complexity analyses of enumeration problems.
In analogue to classical enumeration complexity, there are the classes $\delayFPT$ and $\incFPT$, and, here as well, it is unknown if $\delayFPT\subsetneq\incFPT$ is true or not.

In their research, Creignou et~al.~\cite{ckmmov15} noticed that for some problems, enumerating solutions by increasing size is possible with $\delayFPT$ and exponential space (such as triangulations of graphs, or cluster editings). 
However, it is not clear how to circumvent the unsatisfactory space requirement.
Recently, Meier and Reinbold \cite{DBLP:conf/foiks/0001MR18} observed a similar phenomenon. 
They study the enumeration complexity of problems in a modern family of logic of dependence and independence.
In the context of dependence statements, single assignments do not make sense.
As a result, one introduces \emph{team semantics} which defines semantics with respect to sets of assignments, which are commonly called \emph{teams}.
Meier and Reinbold showed that in the process of enumerating satisfying teams for formulas of a specific \emph{Dependence Logic} fragment, it seemed that an $\FPT$ delay required exponential space.
While reaching polynomial space for the same problem, the price was paid by an increasing delay, $\incFPT$, and it was not clear how to avoid the increase of the delay while maintaining polynomial space.
This is a significant question of research and we improve the understanding of this question by pointing out connections to classical enumeration complexity where similar phenomena have been observed \cite{cs17}.

\paragraph{Related work.} In 1991, Megiddo and Papadimitriou \cite{mp91} introduced the function complexity class $\TF(\NP)$ and studied problems within this class.
In a recent investigation, Goldberg and Papadimitriou introduced a rich theory around this complexity class that features also several aspects of proof theory \cite{gp17}.
Also, the investigations of Capelli and Strozecki \cite{cs17} on probabilistic classes might yield further connections to the enumeration setting via the parametrised analogues of probabilistic computation of Chauhan and Rao \cite{DBLP:conf/caldam/ChauhanR15}.
Furthermore, Fichte et~al.~\cite{DBLP:conf/lata/FichteHS18} study the parametrised complexity of default logic and present in their work a parametrised enumeration algorithm outputting stable extensions. 
It might be worth to further analyse problems in this setting possibly yielding $\incFPT$ algorithms.
Quite recently, Bläsius et~al.~\cite{bfms18} consider the enumeration of minimal hitting sets in lexicographical order and devise some cases which allow for $\delayP$-, resp., $\delayFPT$-algorithms.
Furthermore, there exists a work in which enumeration complexity results have been made for problems on MSO formulas \cite{DBLP:conf/lpar/PichlerRW10}.
Finally, investigations of Mary and Strozecki \cite{ms16} are related to the $\incP$-versus-$\delayP$ question from the perspective of closure operations.

\paragraph{Contribution.}
We improve the understanding of incremental enumeration time by connecting classical enumeration complexity to the very young field of parametrised enumeration complexity.
Although we cannot answer the aforementioned time-space-tradeoff question in either positive or negative way, the presented ``bridge'' to parametrised enumeration will be helpful for future research.
Capelli and Strozecki \cite{cs17} distinguish two kinds of incremental polynomial time enumeration, which we later will call $\incP$ and $\capincP$.
Essentially, the difference of these two classes lies in the perspective of the delay.
For $\incP$ one measures the delay between an output solution $i$ and $i+1$ which has to be polynomial in $i$ and the input length. 
For $\capincP$ the output of $i$ solutions has to be polynomial in $i$ and the input length.
In Section~\ref{sec:connect}, we will introduce several parametrised function classes that are utilised to prove our main result: $\incFPT=\outputFPT$ if and only if $\incP=\outputP$.
This is the first result that directly connects the classical with the parametrised enumeration setting.
By this approach, separating the classical classes then implies separating the parametrised counterparts and \emph{vice versa}.
Moreover, we introduce two hierarchies of parametrised incremental time $\incFPT_a$ and $\capincFPT_a$ in Section~\ref{sec:hierarchy}, show that they interleave and thereby provide some new insights into the interplay of $\FPT$ delay and incremental $\FPT$ delay.
One of the previously mentioned parametrised function classes is a counterpart of the class $\TF(\NP)$, the class of nondeterministic multivalued functions with values that are polynomially verifiable and guaranteed to exist, known from Megiddo and Papadimitriou \cite{mp91}.
This class summarises significant cryptography related problems such as factoring or the discrete logarithm modulo a (certified) prime $p$ of a (certified) primitive root $x$ of $p$.
Clearly, parametrised versions of these problems are members in $\TF(\para\NP)$ via the trivial parametrisation $\kappa_{\text{one}}(x)=1$.

\paragraph{Outline.} In Section~\ref{sec:prelims}, we introduce the necessary notions of parametrised complexity theory and enumeration. 
Then, we continue in Section~\ref{sec:hierarchy} to present two hierarchies of parametrised incremental $\FPT$ enumeration classes and study the relation to $\delayFPT$.
Eventually, in Section~\ref{sec:connect}, we introduce several parametrised function classes and outline connections to the parametrised enumeration classes.
There, we connect a collapse of the two function classes $\TF(\para\NP)$ and $\F(\FPT)$ to a collapse of $\outputFPT$ and $\capincFPT$, extend this collapse to $\TF(\NP)$ and $\FP$ (so in the classical function complexity setting), and further reach out for our main result showing $\outputFPT=\incFPT$ if and only if $\outputP=\incP$.
Finally, we conclude and present questions for future research.

\section{Preliminaries}\label{sec:prelims}
Enumeration algorithms are usually running in exponential time as the solution space is of this particular size.
As Turing machines cannot access specific bits of exponentially sized data in polynomial time, one commonly uses the RAM model as the machinery of choice; see, for instance, the work of Johnson et~al.\ \cite{JohnsonPY88}, or more recently, of Creignou et~al.\ \cite{ckpsv17}.
For our purposes, polynomially restricted RAMs (RAMs where each register content is polynomially bounded with respect to the input size) suffice.
We will make use of the standard complexity classes $\P$ and $\NP$.

\subsection{Parametrised Complexity Theory}
We will present a brief introduction into the field of parametrised complexity theory.
For a deeper introduction we kindly refer the reader to the textbook of Flum and Grohe \cite{FlumGrohe06}.

Let $Q\subseteq\Sigma^*$ be a decision problem over some alphabet $\Sigma$. 
Given an instance $\langle x,k\rangle\in Q\times\Sigma^*$, we call $k$ the \emph{parameter's value (of $x$)}.
Often, instead of using tuple notation for instances, one uses a polynomial time computable function $\kappa\colon\Sigma^*\to\Sigma^*$ (the \emph{parametrisation}) to address the parameter's value of an input $x$.
Then, we write $(Q,\kappa)$ denoting the \emph{parametrised problem} (PP).
Often, the codomain of parametrisations is the natural numbers.

\begin{definition}[Fixed-parameter tractable]
	Let $(Q,\kappa)$ be a parametrised problem over some alphabet $\Sigma$. 
	If there exists a deterministic algorithm $A$ and a computable function $f\colon\N\to\N$ such that for all $x\in\Sigma^*$
	\begin{itemize}
		\item $A$ accepts $x$ if and only if $x\in Q$, and
		\item $A$ has a runtime of $O(f(\kappa(x))\cdot|x|^{O(1)})$,
	\end{itemize}
	then $A$ is an \emph{fpt-algorithm for $(Q,\kappa)$} and $(Q,\kappa)$ is \emph{fixed-parameter tractable} (or short, in the complexity class $\FPT$).
\end{definition}


Flum and Grohe \cite{FlumGrohe06} provide a way to ``parametrise'' a classical and \emph{robust} complexity class.
For our purposes $\para\NP$ suffice and accordingly we do not present the general scheme.
\begin{definition}[$\para\NP$, {\cite[Def.~2.10]{FlumGrohe06}}]
	Let $(Q,\kappa)$ with $Q\subseteq\Sigma^*$ be a parametrised problem over some alphabet $\Sigma$. 
	We have $(Q,\kappa)\in\para\NP$ if there exists a computable function $f\colon\N\to\N$ and a nondeterministic algorithm $N$ such that for all $x\in\Sigma^*$, $N$ correctly decides whether $x\in Q$ in at most $f(\kappa(x))\cdot p(|x|)$ steps, where $p$ is a polynomial.
\end{definition}

Furthermore, Flum and Grohe characterise the class $\para\NP$ via all problems ``\emph{that are in $\NP$ after precomputation on the parameter}''.
\begin{proposition}[Prop.\ 2.12 in {\cite{FlumGrohe06}}]\label{prop:precomputation}
	Let $(Q,\kappa)$ be a parametrised problem over some alphabet $\Sigma$. 
	We have $(Q,\kappa)\in\para\NP$ if there exists a computable function $\pi\colon\Sigma^*\to\Sigma^*$ and a problem $Q'\subseteq\Sigma^*\times\Sigma^*$ such that $Q'\in\NP$ and the following is true: for all instances $x\in\Sigma^*$ we have that $x\in Q$ if and only if $(x,\pi(\kappa(x)))\in Q'$.
\end{proposition}

According to the well-known characterisation of the complexity class $\NP$ via a verifier language, one can easily deduce the following corollary which later is utilised to explain Definition~\ref{def:ffptfparanp}.
\begin{corollary}\label{cor:paraNP-verifier}
	Let $(Q,\kappa)$ be a parametrised problem over some alphabet $\Sigma$ and $p$ some polynomial.
	We have $(Q,\kappa)\in\para\NP$ if there exists a computable function $\pi\colon\Sigma^*\to\Sigma^*$ and a problem $Q'\subseteq\Sigma^*\times\Sigma^*\times\Sigma^*$ such that $Q'\in\PTime$ and the following is true: for all instances $x\in\Sigma^*$ we have that $x\in Q$ if and only if there exists a $y$ such that $|y|\leq p(|x|)$ and $(x,\pi(\kappa(x)),y)\in Q'$.
\end{corollary}	

\subsection{Enumeration}
As already motivated in the beginning of this section, measuring the runtime of enumeration algorithms is usually abandoned.
As a result, one inspects the uniformity of the flow of output solutions of these algorithms rather than their total running time.
In view of this, one measures the delay between two consecutive outputs.
Johnson et~al.\ \cite{JohnsonPY88} laid the cornerstone of this intuition in a seminal paper and introduced the necessary tools and complexity notions.
Creignou, Olive, and Schmidt \cite{CreignouOS11,Schmidt09} present recent notions in this framework, which we aim to follow.
In this paper, we only consider enumeration problems which are ``good'' in the sense of polynomially bounded solution lengths and polynomially verifiable solutions (Capelli and Strozecki \cite{cs17} call the corresponding class EnumP).

\begin{definition}[Enumeration problem]\label{def:enumprob}
	An \emph{enumeration problem} (EP) over an alphabet $\Sigma$ is a tuple $E=(Q,\Sol)$, where 
	\begin{enumerate}
		\item $Q\subseteq\Sigma^*$ is the set of instances (recognisable in polynomial time),
		\item $\Sol\colon\Sigma^*\rightarrow \mathcal{P} (\Sigma^*)$ is a computable function such that for all $x\in\Sigma^*$, $\Sol(x)$ is a finite set and $\Sol(x)\ne\emptyset$ if and only if $x\in Q$,
		\item $\left\{(x,y)\mid y\in\Sol(x)\right\}\in\P$, and
		\item there exists a polynomial $p$ such that for all $x\in Q$ and $y\in\Sol(x)$ we have $|y|\leq p(|x|)$.
	\end{enumerate}
\end{definition}
Furthermore, we use the shorthand $\AllSol=\bigcup_{x\in I}\Sol(x)$ to refer to the set of solutions for every possible instance.
If $E=(Q, \Sol)$ is an EP over the alphabet $\Sigma$, then we call strings $x\in\Sigma^*$ \emph{instances of $E$}, and $\Sol(x)$ the \emph{set of solutions of $x$}. 

An \emph{enumeration algorithm} $\mathcal{A}$ for the enumeration problem $E=(Q, \Sol)$ is a deterministic algorithm which, on the input $x$ of $E$, outputs exactly the elements of $\Sol(x)$ without duplicates, and which terminates after a finite number of steps on every input.	

The following definition fixes the ideas of measuring the `flow of output solutions'.
\begin{definition}[Delay]
 Let $E=(Q,\Sol)$ be an enumeration problem and $\mathcal A$ be an enumeration algorithm for $E$.
 For $x\in Q$ we define the $i$-th delay of $\mathcal A$ as the time between outputting the $i$-th and $(i+1)$-st solution in $\Sol(x)$. 
 Furthermore, we set the $0$-th delay to be the \emph{precomputation phase} which is the time from the start of the computation to the first output statement. 
 Analogously, the $n$-th delay, for $n=|\Sol(x)|$, is the \emph{postcomputation phase} which is the time needed after the last output statement until $\mathcal A$ terminates.
\end{definition}

Subsequently, we will use the notion of delay to state the central enumeration complexity classes.
\begin{definition}
	Let $E=(Q,\Sol)$ be an enumeration problem and $\calA$ be an enumeration algorithm for~$E$. 
	Then $\calA$ is
\begin{enumerate}
	\item an $\enumP$-algorithm if and only if there exists a polynomial $p$ such that for all $x\in Q$, algorithm $\calA$ outputs $\Sol(x)$ in time $O(p(|x|))$.
	\item a $\delayP$-algorithm if and only if there exists a polynomial $p$ such that for all $x\in Q$, algorithm $\calA$ outputs $\Sol(x)$ with delay $O(p(|x|))$.
	\item an $\incP$-algorithm if and only if there exists a polynomial $p$ such that for all $x\in Q$, algorithm $\calA$ outputs $\Sol(x)$ and its $i$-th delay is in $O(p(|x|,i))$ (for every $0\leq i\leq|\Sol(x)|$).
	\item a $\capincP_a$-algorithm if and only if there exists a polynomial $p$ such that for all $x\in Q$, algorithm $\calA$ outputs $i$ elements of $\Sol(x)$ in time $O(p(|x|,i^a))$ (for every $0\leq i\leq|\Sol(x)|$.
	\item a $\totalP$-algorithm if and only if there exists a polynomial $p$ such that for all $x\in Q$, algorithm $\calA$ outputs $\Sol(x)$ in time $O(p(|x|,|\Sol(x)|))$.
\end{enumerate}
Accordingly, we say $E$ is in $\enumP$/$\delayP$/$\incP$/$\capincP_a$/$\totalP$ if $E$ admits an $\enumP$-/$\delayP$-/$\incP$-/$\capincP_a$-/$\totalP$-algorithm.
Finally, we define $\capincP:=\bigcup_{a\in\N}\capincP_a$.
\end{definition}

Note that in the diploma thesis of Schmidt~\cite[Sect.~3.1]{Schmidt09} the class $\enumP$ is called $\complClFont{TotalP}$. 
We avoid this name to prevent possible confusion with class names defined in the following section as well as with the work of Capelli and Strozecki~\cite{cs17}.
Also, we want to point out that Capelli and Strozecki use the definition of $\capincP$ for $\incP$ (and use the name ``\textsf{UsualIncP}'' for $\incP$ instead).  
They prove that the notions of $\capincP$ and $\incP$ are equivalent up to an exponential space requirement when using a structured delay.
So generally, without any space restrictions, the following result applies.
\begin{proposition}[Corollary~13 in \cite{cs17}]\label{prop:capincp=incp}
	$\capincP=\incP$.
\end{proposition}

\subsection{Parametrised Enumeration}
After we introduced the basic principles in parametrised complexity theory and enumeration complexity theory, we will introduce a combined version of these previously introduced notions.

\begin{definition}[\cite{cmmsv17}]\label{def:para-enum-pb}
	A \emph{parametrised enumeration problem} (PEP) over an alphabet $\Sigma$ is a triple $E=(Q,\kappa, \Sol)$ where
	\begin{itemize}
		\item $\kappa\colon \Sigma^*\rightarrow \N$ is a parametrisation (that is, a polynomial-time computable function), and
		\item $(Q, \Sol)$ is an EP.
	\end{itemize}
\end{definition}
Besides, the definitions of enumeration algorithms and delays are easily lifted to the setting of PEPs.

Observe that the following defined classes are in complete analogy to the non-parametrised case from the previous section.

\begin{definition}[\cite{cmmsv17}]\label{def:enum-algs}
 Let $E=(Q, \kappa, \Sol)$ be a PEP and $\mathcal{A}$ an enumeration algorithm for $E$.
 Then the algorithm $\mathcal{A}$ is
\begin{enumerate}
 \item \label{def:FPT-enumerable}
 an $\FPT$-enumeration algorithm if there exist a computable function $t\colon \N\rightarrow \N$ and a polynomial $p$  such that for every instance $x\in\Sigma^*$, $\mathcal{A}$ outputs $\Sol(x)$ in time at most $t(\kappa(x))\cdot p(|x|)$,
 \item a $\delayFPT$-algorithm if there exist a computable function $t\colon \N\rightarrow \N$ and a polynomial $p$ such that for every $x\in\Sigma^*$, $\mathcal{A}$ outputs $\Sol(x)$ with delay of at most $t(\kappa(x))\cdot p(|x|)$,
 \item an $\incFPT$-algorithm if there exist a computable function $t\colon \N\rightarrow \N$ and a polynomial $p$ such that for every $x\in\Sigma^*$, $\mathcal{A}$ outputs $\Sol(x)$ and its $i$-th delay is at most $t(\kappa(x))\cdot p(|x|,i)$, and
 \item an $\outputFPT$-algorithm if there exist a computable function $t\colon \N\rightarrow \N$ and a polynomial $p$  such that for every instance $x\in\Sigma^*$, $\mathcal{A}$ outputs $\Sol(x)$ in time at most $t(\kappa(x))\cdot p(|x|,|\Sol(x)|)$.
\end{enumerate}
\end{definition}

Note that as before, the notion of $\totalFPT$ has been used for the class of $\FPT$-enumerable problems \cite{cmmsv13}.
We avoid this name as it causes confusion with respect to an enumeration class $\complClFont{TotalP}$ \cite[Sect.~3.1]{Schmidt09} which takes into account not only the size of the input but also the number of solutions.
We call this class $\totalP$ instead and accordingly it is the non-parametrised analogue of the above class $\outputFPT$.
Now we group these different kinds of algorithms in complexity classes.

The class $\FPTenum$/$\delayFPT$/$\incFPT$/$\outputFPT$ is the class of all PEPs that admit a $\FPT$-enumeration/$\delayFPT$-/$\incFPT$-/$\outputFPT$-algorithm.  

The class $\delayFPT$ captures a good notion of tractability for parametrised enumeration complexity.
Creignou et~al.\ \cite{cmmsv17} identified a multitude of problems that admit $\delayFPT$ enumeration.
Note that, due to Flum and Grohe~\cite[Prop.~1.34]{FlumGrohe06} the class $\FPT$ can be characterised via runtimes of the form either $f(\kappa(x))\cdot p(|x|)$ or $f(\kappa(x))+ p(|x|)$ (as $a\cdot b\leq a^2+b^2$, for all $a,b\in\N$).
Accordingly, this applies also to the introduced classes $\delayFPT$, $\incFPT$, and $\outputFPT$.

\section{Interleaving Hierarchies of Parametrised Incremental Time}\label{sec:hierarchy}

The previous observations raise the question on how $\delayFPT$ relates to $\incFPT$. 
In the classical enumeration world, this question is answered by Capelli and Strozecki \cite[Prop.~16]{cs17} for the capped version of incremental polynomial time: $\delayP\subsetneq\capincP$ is true, only for linear total time and polynomial space it has not been answered yet.
This is the question how $\delayP$ with polynomial space relates to $\capincP_1$ with polynomial space \cite[Open Problem~1, p.10]{cs17}.
In the course of this chapter, we will realise that the relationship between the classical and the parametrised world is very close.
Capelli and Strozecki approach the separation mentioned above through the classes $\capincP_a$ and prove a strict hierarchy of these classes.
We lift this to the parametrised setting.
\begin{definition}[Sliced Versions of Incremental FPT, extending Def.~\ref{def:enum-algs}]\label{def:slices}
\begin{enumerate}
	\item[3'.] an $\capincFPT_a$-algorithm (for $a\in\N$) if there exists a computable function $t\colon\N\to\N$ and a polynomial $p$ such that for every $x\in\Sigma^*$, $\calA$ outputs $i$ elements of $\Sol(x)$ in time $t(\kappa(x))\cdot i^a\cdot p(|x|)$ (for every $0\leq i \leq |\Sol(x)|$).
	\item[3''.] an $\incFPT_a$-algorithm (for $a\in\N$) if there exists a computable function $t\colon\N\to\N$ and a polynomial $p$ such that for every $x\in\Sigma^*$, $\calA$ outputs $\Sol(x)$ and its $i$-th delay is at most $t(\kappa(x))\cdot i^a\cdot p(|x|)$.
\end{enumerate}
Similarly, we define a hierarchy of classes $\capincFPT_a$ for every $a\in\N$ which consist of problems that admit an $\capincFPT_a$-algorithm.
Moreover, $\capincFPT:=\bigcup_{a\in\N}\capincFPT_a$.
\end{definition}

Clearly, $\bigcup_{a\in\N}\incFPT_a=\incFPT$ and $\incFPT_0=\delayFPT$ by Definition~\ref{def:enum-algs} as the $i$-th delay then merely is $t(\kappa(x))\cdot p(|x|)$, as $i^0=1$.

Agreeing with Capelli and Strozecki \cite[Sect.~3]{cs17}, it seems very reasonable to see the difference of $\incFPT_1$ and $\delayFPT$ anchored in the exponential sized priority queue.
The price of a ``regular'' (that is, polynomial) delay is paid by requiring exponential space.
Though, relaxing this statement shows that the equivalence of incremental $\FPT$ delay and capped incremental $\FPT$-time is also true in the parametrised world.
Similarly, as in the classical setting \cite[Prop.~12]{cs17}, the price of a structured delay is the required exponential space of a priority queue.
\begin{theorem}
	For every $a\ge0$, $\capincFPT_{a+1}=\incFPT_a$.
\end{theorem}
\begin{proof}
	'$\supseteq$': 
	Let $E=(Q,\kappa,\Sol)$ be a PEP in $\incFPT_a$ via an algorithm $\calA$.
	Let $t\colon\N\to\N$ be a computable function and $p\colon\N\to\N$ be a polynomial as in Definition~\ref{def:slices} ({3''.}).
	For every $x\in Q$ algorithm $\calA$ outputs $i$ solutions with a running time bounded by
	\begin{align*}
		\sum_{k=0}^{i}t(\kappa(x))\cdot p(|x|)\cdot k^a =t(\kappa(x))\cdot p(|x|)\cdot\sum_{k=0}^{i}k^a 
		&\leq t(\kappa(x))\cdot p(|x|)\cdot(i+1)\cdot i^a\\
		&\leq 2\cdot t(\kappa(x))\cdot p(|x|)\cdot i^{a+1}.
	\end{align*}
	Accordingly, we have that $E\in\capincFPT_{a+1}$.
	
	'$\subseteq$':
	Now consider a problem $E=(Q,\kappa,\Sol)\in\capincFPT_{a+1}$ via $\calA$ enumerating $i$ elements of $\Sol(x)$ in time $<t(\kappa(x))i^{a+1}\cdot p(|x|)$ for all $x\in Q$, for all $0\leq i\leq |\Sol(x)|$, and some computable function $t$ (see Definition~\ref{def:slices} (3'.)).
	We will show that enumerating $\Sol(x)$ can be achieved with an $i$-th delay of $O(t(\kappa(x))\cdot p(|x|)\cdot q(i)+s)$ where $q(i)=(i+1)^{a+1}-i^{a+1}$ and $s$ bounds the solution sizes (which is polynomially in the input length; w.l.o.g.\ let $p$ be an upper bound for this polynomial).
	To reach this delay, one uses two counters: one for the steps of $\calA$ (\textsf{steps}) and one for the solutions, initialised with value $1$ (\textsf{solindex}).
	While simulating $\calA$, the solutions are inserted into a priority queue $Q$ instead of printing them.
	Eventually the step counter reaches $t(\kappa(x))\cdot p(|x|)\cdot \textsf{solindex}^{a+1}$.
	Then the first element of $Q$ is extracted, output and $\textsf{solindex}$ is incremented by one.
	In view of this, $\calA$ computed \textsf{solindex} many solutions and prints the next one (or $\calA$ already halted).
	Combining these observations leads to calculating the $i$-th delay:
	\begin{align*}
		  &O(t(\kappa(x))\cdot p(|x|)\cdot (i+1)^{a+1} - t(\kappa(x))\cdot p(|x|)\cdot i^{a+1} + s)\\
		=\; &O(t(\kappa(x))\cdot p(|x|)\cdot q(i) +p(|x|))\\
		=\; &O(t(\kappa(x))\cdot p(|x|)\cdot i^a) \qquad (\text{as }q(i)\in O(i^a))
	\end{align*}
	Clearly, this is a delay of the required form $t(\kappa(x))\cdot p(|x|)\cdot i^a$, and thereby $E\in\incFPT_a$.
\end{proof}

Note that from the previous result one can easily obtain the following corollary.

\begin{corollary}\label{cor:capinc=inc}
	$\capincFPT_1=\delayFPT$ and $\capincFPT=\incFPT$.
\end{corollary}

If one drops the restrictions 3.\ and 4.\ from Definition~\ref{def:enumprob}, then Capelli and Strozecki unconditionally show a strict hierarchy for the cap-classes via utilising the well-known time hierarchy theorem \cite{hs65}.
Of course, this result transfers also to the parametrised world, that is, to the same generalisation of $\capincFPT_a$.
Yet it is unknown whether a similar hierarchy can be unconditionally shown for these classes as well as for $\incFPT_a$.
This is a significant question of further research which is strengthened in the following section via connecting parametrised with classical enumeration complexity.

\section{Connecting with Classical Enumeration Complexity} 
\label{sec:connect}

Capelli and Strozecki~\cite{cs17} ask whether a polynomial delay algorithm using exponential memory can be translated into an output polynomial or even incremental polynomial algorithm requiring only polynomial space.
This question might imply a time-space-tradeoff, that is, avoiding exponential space for a $\delayP$-algorithm will yield the price of an increasing $\incP$ delay.
This remark perfectly contrasts with what has been observed by Creignou et~al. \cite{ckmmov15}.
They noticed that outputting solutions ordered by their size seems to require exponential space in case one aims for $\delayFPT$.
As mentioned in the introduction, Meier and Reinbold \cite{DBLP:conf/foiks/0001MR18} observed how a $\delayFPT$ algorithm with exponential space or a specific problem is transformed into an $\incFPT$ algorithm with polynomial space.
These results emphasise why we strive for and why it is valuable to have such a connection between these two enumeration complexity fields.
In this section, we will prove that a collapse of $\incP$ and $\outputP$ implies $\outputFPT$ collapsing to $\incFPT$ and \emph{vice versa}.

Capelli and Strozecki~\cite{cs17} investigated connections from enumeration complexity to function complexity classes of a specific type.
The classes of interest contain many notable computational problems such as integer factoring, local optimisation, or binary linear programming.

It is well known that function variants of classical complexity classes do not contain functions as members but relations instead.
Accordingly, we formally identify languages $Q\subseteq\Sigma^*$ and their solution-space $\AllSol\subseteq\Sigma^*$ with relations $\{(x,y)\mid y\in\Sol(x)\}$ and thereby extend the notation of PPs, EPs, and PEPs.
Nevertheless, it is easy to see how to utilise a \emph{witness function} $f$ for a given language $L$ such that $x\in L$ implies $f(x)=y$ for some $y$ such that $A(x,y)$ is true, and $f(x)=$``no'' otherwise, in order to match the term ``function complexity class'' more adequately.

\begin{definition}
	We say that a relation $A\subseteq\Sigma^*\times\Sigma^*$ is \emph{polynomially balanced} if $(x,y)\in A$ implies that $|y|\leq p(|x|)$ for some polynomial $p$.
\end{definition}

Observe that, for instances of a (P)EP $E$ over $\Sigma$, the length of its solutions are polynomially bounded.
Accordingly, the underlying relation $A\subseteq\Sigma^*\times\Sigma^*$ is polynomially balanced.

The following two definitions present four function complexity classes.

\begin{definition}[$\FP$ and $\FNP$]
Let $A\subseteq\Sigma^*\times\Sigma^*$ be a binary and polynomially balanced relation.
\begin{itemize}
	\item $A\in\F\P$ if there is a deterministic polynomial time algorithm that, given $x\in\Sigma^*$, can find some $y\in\Sigma^*$ such that $A(x,y)$ is true.
	\item $A\in\FNP$ if there is a deterministic polynomial time algorithm that can determine whether $A(x,y)$ is true, given both $x$ and $y$.
\end{itemize}
\end{definition}

\begin{definition}[$\F(\FPT)$ and $\F(\para\NP)$]\label{def:ffptfparanp}
Let $A\subseteq\Sigma^*\times\Sigma^*$ be a parametrised and polynomially balanced problem with parametrisation~$\kappa$.
\begin{itemize}
	\item $A\in\F(\FPT)$ if there exists a deterministic algorithm that, given $x\in\Sigma^*$, can find some $y\in\Sigma^*$ such that $A(x,y)$ is true and runs in time $f(\kappa(x))\cdot p(|x|)$, where $f$ is a computable function and $p$ is a polynomial.
	\item $A\in\F(\para\NP)$ if there exists a deterministic algorithm that, given both $x$ and $y$, can determine whether $A(x,y)$ is true and runs in time $f(\kappa(x))\cdot p(|x|)$, where $f$ is a computable function and $p$ is a polynomial.
\end{itemize}
\end{definition}

Note that the definition of $\F(\para\NP)$ follows the verifier characterisation of precomputation on the parameter as observed in Corollary~\ref{cor:paraNP-verifier}.
Similarly to the classical decision class, $\NP$, the runtime has to be independent of the witness length $|y|$.
\begin{definition}[$\F(\para\NP\cap\para\co\NP)$]
Let $A\subseteq\Sigma^*\times\Sigma^*$ and $B\subseteq\Sigma^*\times\Sigma^*$ be two parametrised and polynomially balanced problems with parametrisations $\kappa$ and $\kappa'$ satisfying the following requirement: for each $x\in\Sigma^*$ either there exists a $y$ with $(x,ay)\in A$, or there is a $z$ with $(x,bz)\in B$, where $a\neq b$ are two special markers in $\Sigma$.
We say that, $(A,B)\in\F(\para\NP\cap\para\co\NP)$ if there exists a nondeterministic algorithm that, given $x\in\Sigma^*$, can find a $y$ with $A(x,ay)$ or a $z$ with $B(x,bz)$ in time $f(\kappa(x))\cdot p(|x|)+g(\kappa'(x))\cdot q(|x|)$, where $p,q$ are polynomials and $f,g$ are computable functions.
\end{definition}

\iflong
\begin{table}
	\centering\small
	\begin{tabular}{lcp{2.5cm}p{4.2cm}}\toprule
		Class & machine & runtime & constraints\\\midrule
		$\FP$ & det. & $p(|x|)$ & find $y$ s.t.\ $A(x,y)$\\
		$\FNP$ & nond. & $p(|x|)$ & guess $y$ s.t.\ $A(x,y)$\\
		$\TF(\NP)$ & nond. & $p(|x|)$ & guess $y$ s.t.\ $A(x,y)$,\newline $A$ is total\\
		$\F(\FPT)$ & det. & $f(\kappa(x))\cdot p(|x|)$ & $\kappa$ parametrisation,\newline find $y$ s.t.\ $A(x,y)$\\
		$\F(\para\NP)$& nond. & $f(\kappa(x))\cdot p(|x|)$ & $\kappa$ parametrisation,\newline guess $y$ s.t.\ $A(x,y)$\\
		$\TF(\para\NP)$& nond. & $f(\kappa(x))\cdot p(|x|)$ & $\kappa$ parametrisation,\newline guess $y$ s.t.\ $A(x,y)$,\newline $A$ is total\\
		$\F(\para\NP\cap\para\co\NP)$ & nond. & $f(\kappa(x))\cdot p(|x|)+g(\kappa(x)')\cdot q(|x|)$ & relations $A,B$ with parametrisations $\kappa$ and $\kappa'$, either find $y$ with $A(x,ay)$ or $z$ with $B(x,bz)$\\
		\bottomrule
	\end{tabular}
	\caption{Overview of function complexity classes. In the machine column `det.'/`nond.' abbreviates `deterministic'/`nondeterministic'. In the runtime column $p$ and $q$ are polynomials, $f$ and $g$ are two computable functions, $\kappa$ is the parameter, and $x$ is the input.}
\end{table}
\fi

In 1994, Bellare and Goldwasser \cite{bg94} investigated functional versions of $\NP$ problems.
They observed that under standard complexity-theoretic assumptions (deterministic doubly exponential time is different from nondeterministic doubly exponential time) these problems are not self-reducible.
Bellare and Goldwasser noticed that these functional versions are harder than their corresponding decision variants.

A binary relation $R\subseteq\Sigma^*\times\Sigma^*$ is said to be \emph{total} if for every $x\in\Sigma^*$ there exists a $y\in\Sigma^*$ such that $(x,y)\in R$.

\begin{definition}[Total function complexity classes]\label{def:totalfunctioncomplexityclasses}
The class $\TF(\NP)$, resp., $\TF(\para\NP)$, is the restriction of $\F\NP$, resp., $\F(\para\NP)$, to total relations.
\end{definition}
The two previously defined classes are \emph{promise classes} in the sense that the existence of a witness $y$ with $A(x,y)$ is guaranteed.
Furthermore, defining a class $\TF(\P)$ or $\TF(\FPT)$ is not meaningful as it is known that $\FP\subseteq\TF(\NP)$ (see, e.g., the work of Johnson et~al.~\cite[Lemma 3]{DBLP:journals/jcss/JohnsonPY88} showing that $\FP$ is contained in $\complClFont{PLS}$, polynomial local search, which is contained in $\TF(\NP)$ by Megiddo and Papdimitriou \cite[p.~319]{mp91}).
Similar arguments apply to $\F(\FPT)\subseteq\TF(\para\NP)$).

Now, we can define a generic (parametrised) search and a generic (parametrised) enumeration problem.  
Note that the parameter is only relevant for the parametrised counterpart named in brackets.
\funcparaproblemdef
{($\para$)$\anotherSol_A$, where $A\subseteq\Sigma^*\times\Sigma^*$}
{$x\in\Sigma^*$, $S\subseteq\Sigma^*$}
{$\kappa\colon\Sigma^*\to\N$}
{output $y$ in $Sol(x)\setminus S$, or answer $S\supseteq \Sol(x)$}
 
\enumproblem
{($\para$)$\enum A$, where $A\subseteq\Sigma^*\times\Sigma^*$}
{$x\in\Sigma^*$}
{$\kappa\colon\Sigma^*\to\N$}
{output all $y$ with $A(x,y)$}

The two results of Capelli and Strozecki \cite[Prop.~7 and 9]{cs17} which are crucial in the course of this section are restated in the following.

\begin{proposition}[Prop.~7 in {\cite{cs17}}]\label{prop:anothersola-in-fp-iff-enuma-in-incp}
	Let $A\subseteq\Sigma^*\times\Sigma^*$ be a binary relation such that, given $x\in\Sigma^*$, one can find a $y$ with $A(x,y)$ in deterministic polynomial time. $\anotherSol_A\in\FP$ if and only if $\enum A\in\capincP$.
\end{proposition}

\begin{proposition}[Prop.~9 in {\cite{cs17}}]\label{prop:tfnp=fp-iff-outputp=incp}
	$\TF(\NP)=\FP$ if and only if $\outputP=\capincP$.
\end{proposition}

In 1991, Megiddo and Papadimitriou studied the complexity class $\TF(\NP)$ \cite{mp91}.
In a recent investigation, Goldberg and Papadimitriou introduced a rich theory around this complexity class that features also several aspects of proof theory \cite{gp17}.
Megiddo and Papadimitriou prove that the classes $\F(\NP\cap\co\NP)$ and $\TF(\NP)$ coincide.
It is easily lifted to the parametrised setting.
\begin{theorem}\label{thm:F(paraNPcapparacoNP)=TF(paraNP)}
	$\F(\para\NP\cap\para\co\NP)=\TF(\para\NP)$.
\end{theorem}
\iflong
\begin{proof}
	We restate the classical proofs of Megiddo and Papadimitriou \cite{mp91}.
	
	``$\subseteq$'': By definition of the class $\F(\para\NP\cap\para\co\NP)$, either there exists a $y$ with $(x,ay)\in A$, or there is a $z$ with $(x,bz)\in B$.
	As a result, the mapping $(A,B)\mapsto A\cup B$ suffices and $A\cup B$ is total.
	So we just need to guess which of $A$ or $B$ to choose. 
	``$\supseteq$'': the mapping $A\mapsto(A,\emptyset)$ is obvious.
\end{proof}
\fi

For the subsequently lemma (which is the parametrised pendant of Prop.~\ref{prop:anothersola-in-fp-iff-enuma-in-incp}) and theorem we follow the argumentation in the proofs of the classical results (Prop.~\ref{prop:anothersola-in-fp-iff-enuma-in-incp} and \ref{prop:tfnp=fp-iff-outputp=incp}).

\begin{lemma}\label{lem:anotherFPT=>IncFPT}
	Let $A\subseteq\Sigma^*$ be a parametrised problem with parametrisation $\kappa$. Then, 
	$\para\anotherSol_A\in\F(\FPT)$ if and only if $\para\enum A\in\capincFPT$.
\end{lemma}
\iflong
\begin{proof}
	``$\Rightarrow$'': 
	Let $\para\anotherSol_A\in\F(\FPT)$ via some algorithm $\calA$.
	Algorithm~\ref{alg:paraenuma-in-incfpt} shows that $\para\enum A\in\capincFPT$. 
\begin{algorithm}[ht]
	\caption{Algorithm showing $\para\enum A\in\capincFPT$.}\label{alg:paraenuma-in-incfpt}
	$S\leftarrow\emptyset$\;
	\Repeat{$S=A(x)$}{
		$y\leftarrow\calA(x,S)$\;
		$S\leftarrow S\cup\{y\}$\;
		\textbf{print} $y$\;
		}
\end{algorithm}
	The runtime of each step is $f(\kappa(x))\cdot p(|x|,|S|)$ for some polynomial $p$ and some computable function $f$.
	Consequently, this shows that $\para\enum A\in\capincFPT$.
	
	``$\Leftarrow$'': 
	Let $\para\enum A\in\capincFPT$.
	Then, there exists a parametrised enumeration algorithm $\calA$ that, given input $x\in\Sigma^*$, outputs $i\leq\Sol(x)$ elements in a runtime of $f(\kappa(x))\cdot i^a\cdot p(|x|)$ for some computable function $f$, $a\in\N$, and polynomial $p$.

	Now, we explain how to compute $\para\anotherSol_A$ in fpt-time.
	Given $(x,S)$, simulate $\calA$ for $f(\kappa(x))\cdot (|S|+1)^a\cdot p(|x|)$ steps.
	If the simulation successfully halts then $\Sol(x)$ is completely output.
	Just search a $y\in \Sol(x)\setminus S$ or output ``$S\supseteq\Sol(x)$''.
	Otherwise, if $\calA$ did not halt then it did output at least $|S|+1$ different elements.
	Finally, we just compute $A(x)\setminus S$ and print a new element.
\end{proof}
\fi

The next theorem translates the result of Proposition~\ref{prop:tfnp=fp-iff-outputp=incp} from classical enumeration complexity to the parametrised setting.

\begin{theorem}\label{thm:tfparanp=ffpt-iff-outputfpt=incfpt}
	$\TF(\para\NP)=\F(\FPT)$ if and only if $\outputFPT=\capincFPT$.
\end{theorem}
\begin{proof}
	``$\Leftarrow$'':
	Let $A(x,y)\in\TF(\para\NP)$ be a parametrised language over $\Sigma^*\times\Sigma^*$ with parametrisation $\kappa$ and $M$ be the corresponding nondeterministic algorithm running in time $g(\kappa(x))\cdot p(|x|)$ for a polynomial $p$, a computable function $g$, and input $x$.
	Now, define the relation $C\subseteq\Sigma^*\times\big\{y\#w\mid y\in\Sigma^*, w\in\{0,1\}^*\big\}$ such that 
	\begin{align*}
	C(x,y\#w) \text{ if and only if }A(x,y)\text{ and }|w|\leq p(|x|). 
	\end{align*}	
	Then for each $x$ there exists $y\#w$ such that $C(x,y\#w)$ is true by definition of $\TF(\para\NP)$.
	Moreover, via padding, for each $x$, there exist at least $2^{p(|x|)}$ solutions $z$ such that $C(x,z)$ is true; in particular, $z$ is of the form $y\#w$ such that $A(x,y)$ is true.
	By construction, the trivial brute-force enumeration algorithm checking all $y\#w$ is in fpt-time for every element of $\Sol(x)$.
	Accordingly, this gives $\para\enum C\in\outputFPT$ as the runtime for $\outputFPT$ algorithms encompasses $|\Sol(x)|$ as a factor.
	
	Then $\para\enum C\in\capincFPT$ and the first $y\#w$ is output in fpt-time. 
	Since $A$ was arbitrary, we conclude with $\TF(\para\NP)=\F(\FPT)$ (as $\F(\FPT)\subseteq\TF(\para\NP)$ by definition).
	
	``$\Rightarrow$'': Consider a problem $\para\enum A\in\outputFPT$ with $\para\enum A=(Q,\kappa,\Sol)$. 
	For every $x\in Q$ and $S\subseteq\Sol(x)$ let $D((x,S),y)$ be true if and only if either ($y\in \Sol(x)\setminus S$) or ($y=\#$ and $S\supseteq \Sol(x)$).
	Then $D\in\TF(\para\NP)$:
	\begin{enumerate}
		\item $D$ is total by construction,
		\item as $\para\enum A$ is a parametrised enumeration problem, there exists a polynomial $q$ such that for every solution $y\in\Sol(x)$ we have $|y|\leq q(|x|)$, and
		\item finally, we need to show that $D((x,S),y)$ can be verified in deterministic time $f(\kappa(x))\cdot p(|x|,|S|,|y|)$ for a computable function $f$ and a polynomial $p$.
		\begin{description}
		\item[Case $\bm{{y\not=\#}}$:] $D((x,S),y)$ is true if and only if $y\in \Sol(x)\setminus S$. 
		This requires testing whether $y\notin S$ and $y\in \Sol(x)$. 
		Both can be tested in polynomial time: $p(|y|,|S|)$, respectively, $p(|x|)$ which follows from Def.~\ref{def:para-enum-pb} (4.).
		\item[Case $\bm{{y=\#}}$:] $D((x,S),y)$ is true if and only if $S\supseteq\Sol(x)$.
		As $\para\enum A\in\outputFPT$ there is an algorithm $\calA$ outputting $\Sol(x)$ in $f(\kappa(x))\cdot p(|x|, |\Sol(x)|)$ steps.
		Now, run $\calA$ for at most $f(\kappa(x))\cdot p(|x|,|S|)$ steps.
		Then finishing within this steps-bound implies that $\Sol(x)$ is completely generated and we merely check $S\supseteq\Sol(x)$ in time polynomial in $|S|$.
		If $\calA$ did not halt within the steps-bound we can deduce $|\Sol(x)|>|S|$.
		Accordingly, $S\not\supseteq\Sol(x)$ follows and $D((x,S),y)$ is not true.
	\end{description}
	\end{enumerate}

	As $\TF(\para\NP)=\FPT$ is true by precondition, given a tuple $(x,S)$, we either can compute a $y$ with $y\in \Sol(x)\setminus S$ or decide there is none (and then set $y=\#$) in fpt-time.
	Accordingly, $\para\anotherSol_A$ is in $\F(\FPT)$ and, by applying Lemma~\ref{lem:anotherFPT=>IncFPT}, we get $\para\enum A$ is in $\capincFPT$.
	This settles that $\outputFPT=\capincFPT$ and concludes the proof.
\end{proof}

The next theorem builds on previous statements in order to connect a collapse in the classical function world to a collapse in the parametrised function world.

\begin{theorem}\label{cor:tfparanp=ffpt-iff-tfnp=fp}
	$\TF(\para\NP)=\F(\FPT)$ if and only if $\TF(\NP)=\FP$.
\end{theorem}
\begin{proof}
Let us start with the easy direction.

``$\Rightarrow$'':	
		Let $A\subseteq \Sigma^{*}\times \Sigma^{*}$ be a total relation in $\TF(\NP)$. 
		By definition of $\TF(\NP)$ and $\TF(\para\NP)$, $(A,\kappa)\in \TF(\para\NP)$ where $\kappa$ is the trivial parametrisation assigning to each $x\in\Sigma^{*}$ the empty string $\varepsilon$. 
		Since $\TF(\para\NP) = \F(\FPT)$, there exists a computable function $f\colon\Sigma^{*}\to \Sigma^{*}$, a polynomial $p$, and a deterministic algorithm $A$, that, given the input $x\in \Sigma^{*}$, outputs some $y\in\Sigma^{*}$ such that $A(x,y)$ in time $f(\kappa(x))\cdot p(x)$. 
		As $\kappa(x)=\varepsilon$ for each $x$, $A$ runs in polynomial time. 
		Accordingly, we have $A\in\FP$ and thereby $\FP=\TF(\NP)$ as $A$ was chosen arbitrarily. 

	``$\Leftarrow$'': 
	Choose some $B\in\TF(\para\NP)$ via machine $M$ running in time $f(\kappa(x))\cdot p(|x|)$ for a polynomial $p$, a computable function $f$, a parametrisation $\kappa$, and an input $x$.
	By Proposition~\ref{prop:precomputation}, we know that there exists a computable function $\pi\colon\Sigma^*\to\Sigma^*$ and a problem $B'\subseteq\Sigma^*\times\Sigma^*\times\Sigma^*$ such that $B'\in\NP$ and the following is true: for all instances $x\in\Sigma^*$ and all solutions $y\in\Sigma^*$, we have that $(x,y)\in B$ if and only if $\bigl((x,\pi(\kappa(x)),y\bigr)\in B'$.
		
		As $B$ is total, $B'$ is total with respect to the third argument as well.
		It follows by assumption that $B'$ is also in $\FP$ via some machine $M'$ having a runtime bounded by a polynomial $q$ in the input length.
		Accordingly, we define a machine $\widetilde M$ for $B$ which, given input $x\in\Sigma^*$, computes $\pi(\kappa(x))$, then simulates $M'$ on $(x,\pi(\kappa(x)))$, and runs in time
		\begin{align}
			f_\pi(\kappa(x))+ q(|\pi(\kappa(x)|,|x|),\tag{$\star$}\label{runtime}
		\end{align}
		 where $f_\pi\colon\Sigma^*\to\N$ is a computable function that estimates the runtime of computing $\pi(\kappa(x))$.
		 Clearly, Equation \eqref{runtime} is an fpt-runtime witnessing $B\in\F(\FPT)$.
		Accordingly, we can deduce that $\TF(\para\NP)=\F(\FPT)$ as $B$ was chosen arbitrarily.
\end{proof}

Combining Theorem~\ref{thm:tfparanp=ffpt-iff-outputfpt=incfpt} with Theorem~\ref{cor:tfparanp=ffpt-iff-tfnp=fp} and finally Proposition~\ref{prop:tfnp=fp-iff-outputp=incp}, connects the parametrised enumeration world with the classical enumeration world.

\begin{corollary}
	$\outputFPT=\capincFPT$ if and only if $\outputP=\capincP$.	
\end{corollary}

If one does not consider space requirements, we can deduce the following corollary by applying Corollary~\ref{cor:capinc=inc} and Proposition~\ref{prop:capincp=incp}.
\begin{corollary}
	$\outputFPT=\incFPT$ if and only if $\outputP=\incP$.	
\end{corollary}

Now, the observations made by Capelli and Strozecki~\cite{cs17} have directly been transferred to our setting.
Accordingly, for instance, the existence of one-way functions would separate $\outputFPT$ from $\incFPT$ as well.
Also a collapse of $\outputFPT$ to $\capincFPT$ would yield a collapse of $\TF(\NP)$ to $\FP$ (Prop.~\ref{prop:tfnp=fp-iff-outputp=incp}) and as well as of $\P$ to $\NP\cap\co\NP$ (due to $\TF(\NP)=\F(\NP\cap\co\NP)$ \cite{mp91}).

\section{Conclusion and Outlook}
We presented the first connection of parametrised enumeration to classical enumeration by showing that a collapse of $\outputFPT$ to $\incFPT$ implies collapsing $\outputP$ to $\capincP$ and \emph{vice versa}.
While proving this result, we showed equivalences of collapses of parametrised function classes developed in this paper to collapses of classical function classes.
In particular, we proved that $\TF(\para\NP)=\F(\FPT)$ if and only if $\TF(\NP)=\FP$. 
The function complexity class $\TF(\NP)$, which has $\TF(\para\NP)$ as its parametrised counterpart, contains significant cryptography-related problems such as factoring.
Furthermore, we studied a parametrised incremental $\FPT$ time enumeration hierarchy on the level of exponent slices (Def.~\ref{def:slices}) and observed that $\capincFPT_1=\delayFPT$.
Also, an interleaving of the two hierarchies, $\incFPT_a$ and $\capincFPT_a$, has been shown.
The results of this paper underline that parametrised enumeration complexity is an area worthwhile to study as there are deep connections to its classical counterpart.

Future research should build on these classes to unveil the presence of exponential space in this setting and give a definite answer to the observed time-space-tradeoffs.
Also, it would be engaging to study connections from parametrised enumeration to proof theory via the work of Goldberg and Papadimitriou \cite{gp17}.
We want to close with the question whether there exist intermediate natural problems between $\F(\FPT)$ and $\TF(\para\NP)$ which are relevant in some area beyond the trivial parametrisation $\kappa_\text{one}(x)=1$.

\subparagraph*{Acknowledgements}
The author thanks Martin Lück (Hannover) and Johannes Schmidt (Jönkö\-ping) for fruitful discussions.

%
\bibliographystyle{plainurl}
\bibliography{incFPT}

\end{document}